\documentclass[final]{econsocart}

\usepackage{booktabs}
\usepackage[yyyymmdd,24hr]{datetime}
\usepackage{caption}
\usepackage{setspace}
\usepackage{subcaption}
\usepackage{mathtools}
\RequirePackage[colorlinks,citecolor=blue,linkcolor=blue,urlcolor=blue,pagebackref]{hyperref}

\startlocaldefs
\usepackage{enumerate}

\theoremstyle{plain}

\newtheorem{theorem}{Theorem}

\newtheorem{lemma}{Lemma}

\newtheorem{corollary}{Corollary}
\theoremstyle{remark}
\newtheorem{definition}{Definition}

\usepackage{tikz}
\usetikzlibrary{shapes}

\newcommand\encircle[1]{%
	\tikz[baseline=(X.base)] 
	\node (X) [draw, shape=circle, inner sep=0] {\strut #1};}

\newcommand\ensquare[1]{%
	\tikz[baseline=(X.base)] 
	\node (X) [draw, shape=rectangle, inner sep=.1cm] {\strut #1};}

\usepackage{etoolbox}

\newcommand{\da}{\mathtt{DA}} 
\newcommand{\eada}{\mathtt{EADA}} 
\newcommand{\dattc}{\mathtt{DA+TTC}}

\newcommand{\rk}{\textup{rk}}

\theoremstyle{remark}

\AtBeginEnvironment{definition}{\vspace{-2.5ex}}

\renewcommand{\appendixname}{Supplementary Appendix}

\usepackage{enumitem}
\usepackage{amstext} 
\usepackage{array}   
\newcolumntype{C}{>{$}c<{$}} 
\usepackage{multirow}
\usepackage{float}
\usepackage{array}   
\newcolumntype{C}{>{$}c<{$}} 

\usepackage{algorithm}
\usepackage{algpseudocode}

\endlocaldefs
\begin{document}
	\begin{frontmatter}

\title{The Trade-off Between Minimal Instability and\\ Larger Improvements over Deferred Acceptance}

\begin{aug}
\author[id=au1,addressref={add1}]{\fnms{Taylor}~\snm{Knipe} \, \and\, \fnms{Josu\'e}~\snm{Ortega}}

\address[id=add1]{%
\orgname{Queen's University Belfast}}

\end{aug}
\support{Knipe: \href{mailto:tknipe01@qub.ac.uk}{tknipe01@qub.ac.uk}, Ortega: \href{mailto:j.ortega@qub.ac.uk}{j.ortega@qub.ac.uk}.  We acknowledge helpful discussions with Péter Biró and Battal Doğan. T. Knipe acknowledges financial support for doctoral studies from Northern Ireland's Department for the Economy.\\
Submitted on {\ddmmyyyydate\today}.}

\begin{abstract}
The celebrated Efficiency-Adjusted Deferred Acceptance mechanism (EADA) improves the efficiency of the DA algorithm via consented priority violations. 
Notwithstanding its many merits, we show that EADA can improve only two students when an alternative mechanism that Pareto-dominates DA could benefit all but one student.
This shortfall in the number of students improved is not exclusive of EADA but extends to all setwise minimally unstable mechanisms, i.e. those that generate a set of blocking pairs that is never a strict superset of that of another mechanism. 
The incompatibility between number of students improved and minimal instability disappears when blocking pairs are compared cardinally rather than by set inclusion. In some problems, EADA can be doubly dominated: improving fewer students while generating more blocking pairs.
\end{abstract}



\end{frontmatter}
\setcounter{footnote}{0}
\newpage

\section{Introduction}
\label{sec:introduction}
The student-proposing Deferred Acceptance (DA) algorithm has become a cornerstone of market design, combining elegant theoretical properties with widespread adoption.
However, while DA generates stable matchings, it can produce allocations that are Pareto-inefficient for students.
The large size of this inefficiency in practice has motivated the development of mechanisms that Pareto-dominate DA by implementing mutually beneficial trades among students, violating only the priorities of students who willingly consent to such violations.\footnote{\cite{abdulkadirouglu2009strategy} report that 4,296 students in NYC would be assigned to a better school without harming any other pupil.}
The most prominent among these is the Efficiency-Adjusted Deferred Acceptance mechanism (EADA, \citealt{kesten2010school}), which addresses DA’s inefficiency while preserving weak but meaningful stability and incentive properties. A large literature has shown that EADA performs remarkably well both theoretically and in the lab, making it a key mechanism in the market design literature and leading to its consideration for practical implementation (see Section \ref{sec:literature}).

In this paper, we study how many students EADA actually improves compared to other mechanisms that also Pareto-dominate DA, and whether these mechanisms improve more students while also generating less instability (as any mechanism that Pareto-dominates DA is bound to create some blocking pairs).

Despite EADA's excellent performance with regard to other criteria, we find that it can perform poorly when judged solely by the number of students it improves. In school choice problems with $n$ students, we construct examples where $n-1$ students could benefit and yet EADA executes only a single two-way exchange instead (even when all students consent to waive their priorities). Therefore, EADA's improvement ratio, defined as the size of the largest improvement possible divided by the size of the actual improvement implemented in the worst-case instance, is $\frac{n-1}{2}$, which is the largest possible ratio obtained by any Pareto-efficient mechanism that improves upon DA and grows linearly with $n$ (Theorem \ref{prop:ratio}). 

Moreover, the concern that only a few students may improve upon their DA placement is not unique to EADA. In fact, any setwise minimally unstable mechanism that Pareto-dominates DA shares the same limitation (Theorem \ref{thm:impo}). The impossibility of improving more students while minimizing justified envy in a setwise sense is particularly relevant because each blocking pair represents an unfairly treated student who could potentially challenge the allocation legally. Therefore, mechanism designers who insist on setwise minimizing envy while improving upon DA must accept that the distribution of improvements may be rather limited compared to what other mechanisms can achieve. Conversely, policymakers who intend to maximize the coverage of improvement across the student population must recognize that they may create a strictly setwise larger number of students who are treated unfairly and thus can challenge the assignment.

Taken together, Theorems \ref{prop:ratio} and \ref{thm:impo} extend Kesten's observation that EADA may not improve as many students as possible by establishing exact bounds on this shortfall and proving an impossibility theorem that highlights the underlying tradeoff between improving more students and being setwise minimally unstable, a property satisfied by EADA among other mechanisms \citep{dougan2021minimally,tang2021weak}.

However, if we compare the instability generated  simply in terms of the number of blocking pairs created, we obtain room for improvement in both number of students improved and number of blocking pairs created. We show that EADA  can be doubly dominated—simultaneously improving fewer students and generating more blocking pairs than alternative mechanisms that also Pareto-dominate DA (Theorem \ref{prop:prop2}). 
This previously unknown bi-dimensional possibility of improvement opens an avenue of research into the relatively under-explored class of mechanisms that Pareto-dominate the ubiquitously used Deferred Acceptance mechanism.

\paragraph{Outline.}
The remainder of this paper is organised as follows.
Section \ref{sec:literature} discusses the literature.
Section \ref{sec:model} introduces the model.
Section \ref{sec:limitations} presents our results.
Section \ref{sec:conclusion} concludes.

\section{Related Literature}
\label{sec:literature}

\paragraph{DA's Inefficiency.}
The Pareto-inefficiency of the student-proposing Deferred Acceptance algorithm has been recognized since the early studies on matching markets and school choice \citep{roth1982economics, abdulkadirouglu2003}. It has been consistently documented in practice  \citep{abdulkadiroglu2005,abdulkadirouglu2009strategy,che2019efficiency,aeri,ortega2023cost}. 

\paragraph{Mechanisms That Improve Upon DA.}
The most prominent solution to DA's inefficiency is the Efficiency-Adjusted Deferred Acceptance mechanism, proposed in a seminal paper by \cite{kesten2010school}. EADA identifies Pareto-improving trading cycles among students that only violate the priorities of students who benefit from consenting to waive their priorities, ensuring such students have no incentive to withhold consent. EADA has been extensively studied both theoretically and experimentally, with the literature consistently documenting its strong performance across multiple stability and incentives criteria. 

From this literature, the papers closest to ours are \cite{dougan2021minimally,tang2021weak} and \cite{kwon2020justified}. These three papers show that EADA is setwise minimally unstable when using blocking pairs or triplets for the comparison, respectively.\footnote{Outside of the class of mechanisms that Pareto-dominate DA, \cite{aeri} and \cite{ehlers2021robust} show that the Top Trading Cycles algorithm is also setwise minimally unstable.} We contribute to this literature by showing that any setwise minimally unstable mechanism is bound to improve only a few students when many could alternatively benefit. 

Furthermore, \cite{dougan2021minimally} also show that when comparing minimal instability in a cardinal sense, there is no Pareto-efficient mechanism that Pareto-dominates DA that is cardinally minimally unstable. In this direction, we further show that it is possible to generate fewer blocking pairs than EADA while improving more students.

Beyond these papers, the broader literature has consistently praised EADA's performance across other dimensions  \citep{bando2014existence, tang2014new, dur2019school, troyan2020essentially,troyan2020obvious,ehlers2020legal,tang2021weak,reny2022efficient,chen2023regret,cerrone2022school,cheng4961755equilibrium,dogan2023existence,shirakawa2024simple}.	

Another natural approach to Pareto-improve DA is to apply the Top Trading Cycles mechanism to DA's allocation. While this mechanism is known in the literature, it has received surprisingly limited formal analysis (with some exceptions, see \citealt{alcalde2017,kesten2010school,troyan2020essentially}). All of our EADA results extend to DA+TTC: it can miss significantly larger improvements while generating more blocking pairs.

More generally, a series of papers have examined when it is possible to improve upon DA while maintaining strategy-proofness \citep{kesten2019strategy} or consistency \citep{dougan2020consistent} . However, these Pareto improvements have inherent limits: they cannot address DA's poor rank distributions or student segregation \citep{ortega2025pareto}.

\paragraph{Scope of Efficiency Gains.}
\cite{ortega2024unimprovable} show that DA is Pareto-inefficient with high probability in large i.i.d. random markets and that most students can potentially benefit from disjoint trading cycles. Furthermore,  all efficient mechanisms that dominate DA tend to improve nearly every student.
These findings provide context for our analysis. While their paper shows that improvement mechanisms converge in their scope in large i.i.d. random markets, our paper demonstrates that this equivalence does not hold in specific instances, showing that the number of students who benefit can vary dramatically between different efficient mechanisms that dominate DA.

\section{Model}
\label{sec:model}
A school choice problem $P$ consists of a set of $n$ students $I$ and a set of schools $S$. Each student $i \in I$ has a strict preference relation $\succ_i$ over the schools. Each school $s \in S$ has a quota of available seats $q_s$ and a strict priority relation $\triangleright_s$ over the students. The school $s_\emptyset$ with $q_{s_\emptyset}=\infty$ denotes being unassigned.
For a given school choice problem $P$, a \emph{matching} $\mu$ is a mapping from $I$ to $S$ such that no school is matched to more students than its quota. We denote by $\mu_i$ the school to which student $i$ is assigned in $\mu$.

The function $\rk_i:S \rightarrow \{1, \ldots, n\}$ specifies the rank of school $s$ according to the preference relation $\succ_i$ of student $i$:
\begin{align}
\rk_i(s) = |\{s' \in S: s' \succ_i s\}|+1,
\end{align}
so that student $i$'s most preferred school gets a rank of $1$. With some abuse of notation, we use the same rank function to specify the students' rank per the priority profile of schools.

A matching $\mu$ \textit{weakly Pareto-dominates} matching $\nu$ if, for every student $i \in I$, $\rk_i(\mu_i) \leq \rk_i (\nu_i)$. A matching $\mu$ \textit{Pareto-dominates} matching $\nu$ if $\mu$ weakly dominates $\nu$ and there exists a student $j \in I$ with $\rk_j(\mu_j) < \rk_j (\nu_j)$. A matching is \textit{Pareto-dominated} if there exists a matching that Pareto-dominates it and is \textit{Pareto-efficient} if it is not Pareto-dominated. 

Student $i$ \emph{desires} school $s$ in matching $\mu$ if $\rk_i(s)< \rk_i(\mu_i)$ and he \emph{envies} student $j$ at matching $\mu$ if $\rk_i(\mu_j)<\rk_i(\mu_i)$. We say that student $j$ \emph{violates} student $i$'s priority at school $s$ in matching $\mu$ if $i$ desires $s$, $\mu_j =s$, and $\rk_s(i) < \rk_s(j)$. 
In this case, we say that student $i$'s envy towards $j$ is justified, and we call ($i,j,s$) and ($i,s$) a blocking triplet and a blocking pair in matching $\mu$, respectively. We use $B(\mu)$ to denote the set of blocking pairs in matching $\mu$.
A matching $\mu$ is \emph{non-wasteful} if every school $s$ that is desired by some student in $\mu$ satisfies $|\{ i \in I:\mu_i =s\}| =q_s$. A matching $\mu$ is \textit{stable} if it is non-wasteful and no student’s priority at any school is violated in $\mu$.\\

\emph{Mechanisms with and without Consent.} A \emph{mechanism} is a function that maps every school choice problem to a matching. 
We focus on two well-known mechanisms. The first of these is the student-proposing Deferred Acceptance mechanism (DA, \citealt{gale1962}), which works as follows. Students apply to their most preferred school. Schools tentatively accept applicants up to their capacity, rejecting the rest. Rejected students apply to their next choice in subsequent rounds. The algorithm terminates when no new rejections occur, yielding the student-optimal stable matching $\da(P)$.

The second mechanism of interest is Top Trading Cycles on top of DA (DA+TTC). This mechanism first computes DA, then applies Gale's Top Trading Cycles algorithm where each student initially ``owns'' their DA assignment \citep{shapley1974}. Students point to owners of their most preferred schools, and all students in cycles trade. The process repeats until no cycles remain.

Standard mechanisms simply return a matching for any school choice problem. A \emph{consent-based mechanism} $M$ instead maps every school choice problem \emph{and} a subset of students $W \subseteq I$ to a matching. 
The interpretation of the subset of students $W$ is that these students agree to have their priorities violated at some schools as long as consenting to waive their priorities does not harm them, and thus we refer to $W$ as the \emph{consent structure}. Formally, the property that waiving their priorities never harms consenting students means that for any problem $P$ and any consent structure $W \subseteq I \setminus \{i\}$:
$$\rk_i[M(P,W \cup \{i\})]\leq \rk_i[M(P,W )]$$

Note that for any fixed consent structure $W$, a consent-based mechanism $M(\cdot,W)$ reduces to a standard mechanism that maps problems to matchings. 
Efficiency-Adjusted Deferred Acceptance (EADA, \citealt{kesten2010school}) is an example of a consent-based mechanism. To describe EADA, we first define an \emph{interrupter}, which is a student-school pair $(i,s)$ such that:
\begin{itemize}
	\item student $i$ is tentatively accepted at $s$ in step $t$,
	\item student $j\neq i$ is rejected at $s$ in step $t'\geq t$,
	\item student $i$ is rejected from school $s$ in step $t''\geq t'$
\end{itemize}

We say that an interrupter $(i,s)$ consents if $i \in W$.
Having defined the notion of an interrupter, we now describe the celebrated (and slightly more complex) EADA consent-based mechanism in detail below.

\begin{algorithm}[H]
	\caption{Efficiency-Adjusted Deferred Acceptance (EADA)}
	\begin{algorithmic}[1]
		\State \textbf{Input:} Problem $P$ and consent structure $W \subseteq I$
		\State Compute initial matching $\mu^0 = \da(P)$
		\While{consenting interrupters exist}
		\State Find the last round $t$ where	a consenting interrupter $(i,s)$ was rejected
		\State Identify all consenting interrupters $(i,s)$ in round $t$
		\If{no consenting interrupters found}
		\State \textbf{break}
		\EndIf
		\For{each consenting interrupter $(i,s)$ in round $t$}
		\State Remove school $s$ from student $i$'s preferences
		\EndFor
		\State Rerun DA with modified preferences
		\EndWhile
		\State \Return final matching
	\end{algorithmic}
\end{algorithm}

We denote by $\eada(P,W)$ the matching generated by Efficiency-Adjusted Deferred Acceptance algorithm in school choice problem $P$ with consent structure $W$. We refer to the mechanism $\eada(\cdot,I)$ as EADA with full consent.

EADA with full consent can be more straightforwardly computed by identifying every school who did not reject any student in DA's execution (these schools are called \emph{under-demanded}). After finding DA's allocation, those schools are permanently matched with their assigned students, and such under-demanded schools and corresponding students are removed from DA. Then DA is executed in the simplified problem, and this procedure is repeated until all schools have been removed. This procedure was proposed by \cite{tang2014new} and can be extended to partial consent structures, but this will not be necessary for our purposes.

\paragraph{Setwise and Cardinally Minimally Unstable Mechanisms.} Since it is well-known that every mechanism that Pareto-dominate DA generates blocking pairs in general, we need a way to compare their relative performance in terms of instability (e.g. blocking pairs). We use the notion of setwise minimally unstable mechanisms, a popular way of comparing the blocking pairs generated by assignment mechanisms \citep{tang2021weak, dougan2021minimally,aeri,kwon2020justified,bonkoungou2025reforms}.\\
\begin{definition}
A Pareto-efficient mechanism $M$ is \emph{setwise minimally unstable} if there does not exist any problem $P$ and any other Pareto-efficient mechanism $N$ such that 
$B(N(P)) \subsetneq B(M(P))$.
\end{definition}

EADA with full consent satisfies this property \citep{dougan2021minimally,tang2021weak}, as do many other mechanisms, since the setwise criterion creates only a partial order across mechanisms. Therefore, mechanisms that fail to be setwise minimally unstable are quite disappointing from a legal and fairness perspective, as they generate an unnecessarily large number of students who have their priorities violated and who may legally challenge their allocation. However, as we will demonstrate, this desirable stability property comes at a significant cost in terms of the number of students who can be improved.

A different way to compare the instability among mechanisms that dominate DA is to count the number of blocking pairs they generate, rather than comparing them by set inclusion. This leads to a notion of cardinally more stable mechanisms. 
\begin{definition}
	A Pareto-efficient mechanism $M$ is \emph{cardinally more stable} than a Pareto-efficient mechanism $N$ if, for every problem $P$,	$|B(M(P))| \leq |B(N(P))|$, with strict inequality for some $P$.
\end{definition}

We also define two other useful properties of mechanisms used throughout the paper. 
First, we say that a mechanism $M$ \emph{Pareto-dominates DA} if: (i) for any problem $P$, the matching $M(P)$ weakly Pareto-dominates the matching $\da(P)$, and (ii) there exists some problem $P$ where the matching $M(P)$ strictly Pareto-dominates the matching $\da(P)$.
Second, we say that a mechanism $M$ is Pareto-efficient if, for any problem $P$, the matching $M(P)$ is Pareto-efficient.

\paragraph{Envy Digraphs and Improvement Cycles.}
To analyse improvements over DA, we adopt a graph-theoretic framework. Formally,
the \emph{envy digraph} $G^{\da(P)}$ is a directed graph where nodes represent students  and a directed edge $(i,j)$ exists if student $i$ envies student $j$'s assignment under DA.We denote the edge $(i,j)$ as $i \rightarrow j$. 

A \emph{cycle} $C$ in the envy digraph is a sequence of nodes $C=(i_0 \to i_1 \to \ldots \to i_k \to i_0)$ such that there is a directed edge between each pair of consecutive nodes and no edge is repeated. A \emph{trading cycle} is a cycle in which every node appears exactly once, except for $i_0$. Unless specifically stated otherwise, every cycle mentioned henceforth is a trading cycle. 

We focus on school choice problems where $G^{\da(P)}$ contains at least one cycle.
A \emph{feedback set} $F$ is a collection of pairwise-disjoint cycles in $G^{\da(P)}$ such that, when all nodes in these cycles are removed (together with their adjacent edges), the resulting subgraph is acyclic.\footnote{Feedback sets, sometimes called feedback vertex sets, have been extensively studied in graph theory \citep{karp2010reducibility}. }
Feedback sets represent sets of disjoint trading cycles that can be simultaneously implemented to improve upon DA. 

We use $V(F)$ to denote the set of nodes (students) that are included in some cycle in $F$. If $i \in V(F)$, we say that $F$ covers $i$. 
Given a feedback set $F$, we define the matching $\mu^F$ as the allocation obtained when the trades in each cycle of $F$ are implemented, starting from DA.

Our analysis builds upon a well-known result which connects improvements over DA to cycles in envy digraphs:

\begin{lemma}[\cite{tang2014new}]\label{lemma:imporvementsaretrades} Every Pareto improvement over DA corresponds to a set of pairwise-disjoint trading cycles in $G^{\da(P)}$.\end{lemma}

Having established our framework, we are now ready to state our results.

\section{Limitations of Existing Mechanisms}
\label{sec:limitations}
How many students do existing mechanisms actually improve compared to the maximum improvement possible over DA? To quantify this gap, we define the improvement ratio of a mechanism. Let $\mathcal{I}(\mu,P)$ denote the set of students improved by the matching $\mu$ that Pareto-dominates DA in problem $P$:
\begin{equation}
\mathcal{I}(\mu,P) \coloneqq \{i \in I : \rk_i(\mu_i)<\rk_i(\da_i(P)) \}
\end{equation} 

Let $\mathcal{M}^\da(P)$ denote the set of matchings that weakly Pareto-dominate DA in $P$.

We define the maximum improvement $\mathcal{I^*}(P)$ as:
\begin{equation}
\mathcal{I^*}(P) \coloneqq \max_{\mu \in \mathcal{M}^\da(P)} |\mathcal{I}(\mu,P)|
\end{equation}

Let $\mathcal{P}_{n}$ denote the class of all school choice problems with $n$ students. For any mechanism $M$ that Pareto-dominates DA, we define its \emph{improvement ratio} as:
\begin{equation}
\rho(M;n) \coloneqq \max_{P \in \mathcal{P}_{n}} \frac{|\mathcal{I^*}(P)|}{|\mathcal{I}(M(P),P)|}
\end{equation}

For example, an improvement ratio of 2 means that the mechanism that improves as many students as possible compared to DA may generate twice as many improvements as mechanism $M$. Our first result shows that this ratio can grow linearly with the size of the problem, and particularly that the improvement ratio of DA+TTC and EADA with full consent achieve the worst possible improvement ratio among all mechanisms that weakly Pareto-dominate DA.

\begin{theorem}
\label{prop:ratio}
For any $n > 4$ and any Pareto-efficient mechanism $M$ that Pareto-dominates DA:
\begin{equation*}
\rho(\eada;n) = \rho(\dattc;n) = \frac{n-1}{2} \geq \rho(M;n)
\end{equation*}
\end{theorem}

\begin{proof}
{\it Upper Bound.} First we prove that the improvement ratio cannot exceed $\frac{n-1}{2}$ for any mechanism that Pareto-dominates DA.
For the numerator: At most $n-1$ students can improve over DA, as there is always at least one unimprovable student \citep{tang2014new}.
For the denominator: The minimum number of students improved by any Pareto-efficient mechanism that  Pareto-dominates DA is 2, because a trading cycle requires at least two students.\footnote{We have assumed that DA is inefficient; otherwise, the minimum would be 0 and the ratio undefined.}
Therefore, the ratio cannot exceed $\frac{n-1}{2}$.

\noindent {\it Tightness.} To show this bound is achievable, we construct a school choice problem with $n$ students and schools (all with unit capacity) where
EADA with full consent and DA+TTC achieve this improvement ratio while a different matching improves $n-1$ students.

Let school priorities be such that, for every school $s_k$ except $s_1$:
\begin{equation}
s_k: i_k \triangleright_{s_k} i_{k-1}
\end{equation}

Students' preferences are such that, for every student $i_k \in I \setminus \{i_1, i_{n-2}, i_{n-1}, i_n$\}:
\begin{equation}
i_k: s_1 \succ_{i_k} s_{k+1} \succ_{i_k} s_k
\end{equation}

The remaining priorities and preferences are as follows:
\begin{align}
s_1&: i_1 \triangleright_{s_1} i_n \triangleright_{s_1} i_2 \triangleright_{s_1} i_{n-1}\\
i_1 &: s_2 \succ_{i_1} s_{n-1} \succ_{i_1} s_1\\
i_{n-2} &: s_1 \succ_{i_{n-2}} s_2 \succ_{i_{n-2}} s_{n-2}\\
i_{n-1} &: s_1 \succ_{i_{n-1}} s_2 \succ_{i_{n-1}} s_{n-1}\\
i_n &: s_1 \succ_{i_n} s_n
\end{align}

The table below demonstrates this construction for $n=7$. The DA allocation appears in bold. The allocation that maximises the number of improvements appears in circles, whereas the matching recommended by both EADA with full consent and DA+TTC appears in squares (whenever different from DA).
\begin{center}
\begin{tabular}{ccccccc|ccccccc}
\multicolumn{14}{c}{Example 1} \\
\hline
$i_1$ & $i_2$ & $i_3$ & $i_4$ & $i_5$ & $i_6$ & $i_7$ & $s_1$ & $s_2$ & $s_3$ & $s_4$ & $s_5$ & $s_6$ & $s_7$ \\
\hline
\ensquare{$s_2$} & \ensquare{$s_1$} & $s_1$ & $s_1$ & $s_1$ & \encircle{$s_1$} & $s_1$ & $\bm{i_1}$ & $\bm{i_2}$ & $\bm{i_3}$ & $\bm{i_4}$ & $\bm{i_5}$ & $\bm{i_6}$ & $\bm{i_7}$ \\
\encircle{$s_6$} & \encircle{$s_3$} & \encircle{$s_4$} & \encircle{$s_5$} & \encircle{$s_2$} & $s_2$ & $\bm{s_7}$ & $i_7$ & \ensquare{$i_1$} & \encircle{$i_2$} & \encircle{$i_3$} & \encircle{$i_4$} & $i_5$ & $i_6$ \\
$\bm{s_1}$ & $\bm{s_2}$ & $\bm{s_3}$ & $\bm{s_4}$ & $\bm{s_5}$ & $\bm{s_6}$ & $\cdot$ & \ensquare{$i_2$} & $\cdot$ & $\cdot$ & $\cdot$ & $\cdot$ & $\cdot$ & $\cdot$ \\
$\cdot$ & $\cdot$ & $\cdot$ & $\cdot$ & $\cdot$ & $\cdot$ & $\cdot$  & $i_6$  & $\cdot$ &$\cdot$ & $\cdot$ & $\cdot$ & $\cdot$ & $\cdot$ 
\end{tabular}
\end{center}

We now demonstrate that the unique stable matching assigns each student $i_k$ to the school with the corresponding index $s_k$.

\paragraph{Execution of DA.} In DA's round 1, student $i_1$ applies to $s_2$, whereas all other students apply to $s_1$. School $s_1$ temporarily accepts $i_n$ (its highest-priority applicant) and rejects all others.
In round 2, rejected students $i_2$ through $i_{n-1}$ make their next applications. Students $i_2$ through $i_{n-3}$ apply to their second-choice schools $s_3$ through $s_{n-2}$, while $i_{n-2}$ and $i_{n-1}$ both apply to $s_2$. School $s_2$ temporarily accepts $i_1$ (who applied in round 1) and rejects $i_{n-2}$ and $i_{n-1}$.
In subsequent rounds, a cascade of rejections occurs as students apply to their next preferred schools, eventually resulting in each student $i_k$ being matched with school $s_k$.

The uniqueness of the set of stable matchings can be verified by observing that the school-proposing DA algorithm yields the same outcome, with each school immediately being matched to its highest-priority student.

\paragraph{Analysis of EADA with full consent.} We begin by identifying under-demanded schools. First, $s_n$ is under-demanded and is removed together with $i_n$. In the next iteration, $s_3, s_4, \ldots, s_{n-1}$ are under-demanded and removed jointly with $i_3, i_4, \ldots, i_{n-1}$, respectively. At this stage, only schools $s_1$ and $s_2$ remain with students $i_1$ and $i_2$. These students form a trading cycle where $i_1$ improves to $s_2$ and $i_2$ improves to $s_1$. Thus, EADA with full consent improves exactly 2 students (assuming $i_n$ consents to waive his priority at $s_1$, otherwise no student improves).

\paragraph{Analysis of DA+TTC.} In DA+TTC, we begin with the DA allocation where each student $i_k$ is assigned to school $s_k$. In the TTC algorithm, only students $i_1$ and $i_2$ form a trading cycle, as $i_1$ points to $i_2$ and $i_2$ points to $i_1$. After executing this trade, students $i_3$ through $i_n$ remain at their DA allocations which are preferred to any remaining allocation. Therefore, DA+TTC also improves exactly 2 students.

\paragraph{Maximum Improvement.} There exists, however, a matching $\mu^*$ that dominates DA, obtained by implementing the cycles $(i_1 \rightarrow i_{n-1} \rightarrow i_1)$ and $(i_2 \rightarrow i_3 \rightarrow \ldots \rightarrow i_{n-2} \rightarrow i_2)$, which improves $n-1$ students. Under this matching, student $i_n$ receives the same allocation as under DA, while all other students receive a more preferred school.

Therefore, the improvement ratio equals $\frac{n-1}{2}$ for any $n > 4$, establishing the tightness of our bound.
\end{proof}

Some remarks are now in order. First, \cite{kesten2010school} already pointed out in his seminal paper that larger improvements than EADA were possible. Our contribution here is to quantify how much larger those improvements can be even when all students consent to waive their priorities, leading to an improvement gap that grows linearly with the problem size.

Second, in our example, the maximum improvement comes at the cost of generating more blocking pairs. In the particular example with seven agents, the larger improvement generates at least three blocking pairs ($\{i_1,s_2\},\{i_2,s_1\},\{i_7,s_1\}$) versus a single blocking pair generated by EADA with full consent and DA+TTC ($\{i_7, s_1\}$). 
We do not take a stance on whether the larger improvement is worth the additional justified envy. Our point is to make explicit that such reduction in justified envy comes at a high cost in terms of how many students are benefited. 
Furthermore, as we will show in the next subsection, larger improvements can sometimes lead to fewer blocking pairs.

Third, it is natural to ask whether EADA with full consent and DA+TTC are the only mechanism that fails to implement larger improvements that benefit more students. We answer this question in the negative in the next Subsection.

\subsection{A General Impossibility Result}
We now generalize Theorem \ref{prop:ratio} to a wide class of mechanisms beyond EADA with full consent and DA+TTC, namely all those that are setwise minimally unstable, i.e. those which never produce a set of blocking pairs that is strictly larger by set inclusion than that generated by another Pareto-efficient mechanism that Pareto-dominates DA.
\begin{theorem}
	\label{thm:impo}
	Among Pareto-efficient mechanisms that Pareto-dominate DA, every setwise minimally unstable mechanism achieves an improvement ratio of $\frac{n-1}{2}$.
\end{theorem}

\begin{proof}
	We use the same problem as in the proof of Theorem \ref{prop:ratio}. In the case with seven students and schools, there are six cycles in DA's envy digraph, summarized in Table \ref{tab:cycles} below.
	\begin{table}[h!]
		\centering
		\caption{Cycles and Blocking Pairs in Example 1}
		\label{tab:cycles}
		\begin{tabular}{ll}
			\toprule
			\bf Cycles & \bf Blocking Pairs\\
			\toprule
			$C_1=(i_1 \to i_2 \to i_1)$ & $(i_7, s_1)$\\
			$C_2=(i_1 \to i_2 \to i_3 \to  i_1)$ & $(i_7, s_1),(i_6, s_1), (i_2,s_1)$\\
			$C_3=(i_1 \to i_2 \to i_3 \to  i_4 \to i_1)$ & $(i_7, s_1),(i_6, s_1), (i_2,s_1)$\\
			$C_4=(i_1 \to i_2 \to i_3 \to  i_4 \to i_5 \to i_1)$ & $(i_7, s_1),(i_6, s_1), (i_2,s_1)$\\
			$C_5=(i_2 \to i_3 \to  i_4 \to i_5 \to i_2)$ & $(i_1, s_2)$\\
			$C_6=(i_1 \to i_6 \to i_1)$ & $(i_7, s_1), (i_2,s_1), (i_5,s_6)$\\
			$C_7=(i_1 \to i_6 \to i_2 \to i_1)$ & $(i_7, s_1), (i_1,s_2), (i_5,s_6)$\\
			\bottomrule
		\end{tabular}
	\end{table}
	
	Any Pareto-efficient mechanism that Pareto-dominates DA must implement at least one of these cycles (potentially two if mutually disjoint). In particular, the larger improvement that includes as many students as possible  implements $C_5=(i_2  \to i_3 \to i_4 \to i_5 \to i_2)$ and $C_6=(i_1 \to i_6 \to i_1)$. 
	
	The cycles $C_2, C_3, C_4, C_6, C_7$ all generate a set of blocking pairs that is strictly larger (by set inclusion) than that generated by $C_1$. The cycle $C_5$ is incomparable to $C_1$, but implementing it by itself does not lead to a Pareto-efficient allocation, since the cycle $C_6$ could be jointly implemented. But if both $C_5$ and $C_6$ are implemented, the set of blocking pairs becomes $\{(i_1,s_2), (i_2,s_1), (i_7,s_1)\}$, a strict superset of the set of blocking pairs generated by $C_1$. Thus, we conclude that every setwise minimally unstable must implement $C_1$, therefore missing the larger improvements dictated by $C_5$ and $C_6$ jointly, leading to an improvement ratio of $\frac{n-1}{2}=\frac{6}{2}$.
	
	Finally, the extension to arbitrary $n$ is straightforward: there are $n$ cycles, with $C_k=(i_1 \to \ldots \to i_{k+1} \to i_1)$ for all integer $1 \leq k \leq n-3$, all of them generating the blocking pair ($i_n,s_1$), and only $C_1$ introducing exactly this blocking pair. 	
	There are exactly three more cycles: $C_{n-2}=(i_2 \to i_3 \to \ldots \to i_{n-3} \to i_2)$, $C_{n-1}=(i_1 \to i_6 \to i_1)$, and $C_n=(i_1 \to i_{n-1} \to i_2 \to i_1)$. 
	Only one collection of cycles that leaves the graph acyclic is setwise minimally unstable: $C_1$. 
	However, the cycles $C_{n-2}$ and $C_{n-1}$ are disjoint (and thus compatible) and lead to an allocation where $n-1$ students improve, and therefore any setwise minimally unstable mechanism incurs in the largest improvement ratio possible, namely $\frac{n-1}{2}$.
\end{proof}

Noting that EADA is setwise minimally unstable (as proven by \cite{tang2021weak} and \cite{dougan2021minimally}) makes it evident that Theorem \ref{thm:impo} implies Theorem \ref{prop:ratio}. Nonetheless, we present both Theorems independently to simplify the argument and improve their readability. 

\subsection{Larger Improvements with Less Justified Envy}
Theorems \ref{prop:ratio} and \ref{thm:impo} shows that EADA with full consent and DA+TTC can improve significantly fewer students compared to the DA baseline, but that such smaller coverage is unavoidable if we insist on mechanisms that are setwise minimally unstable. Now we show that the impossibility disappears if we compare blocking pairs simply by counting them, i.e. we consider cardinal minimal instability.

To be precise, let $M$ and $N$ be Pareto-efficient mechanisms that Pareto-dominate DA. We say that $M$ improves more students than $N$ if for every problem $P$, we have $|\mathcal I(N(P),P)|\leq |\mathcal I(M(P),P)|$, with strict inequality for some $P$. We are now ready to state our result.

\begin{theorem}\label{prop:prop2} 
	There exists a Pareto-efficient mechanism that Pareto-dominates DA that improves more students than EADA with full consent while being cardinally more stable than EADA with full consent. An analogous result holds for DA+TTC.
\end{theorem}

\begin{proof}
{\it Comparison with EADA with full consent.} Consider the following school choice problem with seven students and seven schools with one seat each. DA's allocation appears in bold. Preferences and priorities left unspecified are irrelevant.

\begin{center}
\begin{tabular}{ccccccc|ccccccc}
\multicolumn{14}{c}{Example 2} \\
\hline
$i_1$ & $i_2$ & $i_3$ & $i_4$ & $i_5$ & $i_6$ & $i_7$ 						& $s_1$ & $s_2$ & $s_3$ & $s_4$ & $s_5$ & $s_6$ & $s_7$ \\
\hline
$s_6$ & $s_1$ & $s_6$ & $s_5$ & $s_3$ & $s_4$ & $s_4$ 						& $\bm{i_1}$ & $\bm{i_2}$ & $\bm{i_3}$ & $\bm{i_4}$ & $\bm{i_5}$ & $\bm{i_6}$&$\cdot$ \\
$s_4$ & $\bm{s_2}$ & $\bm{s_3}$ & $\bm{s_4}$ & $s_6$ & $\pmb{s_6}$ & $\pmb{s_7}$& $i_5$ & $i_1$ & $i_5$ & $i_7$ & $i_4$ & $i_3$&$\cdot$ \\
$s_2$ & $\cdot$ & $\cdot$ & $\cdot$ & $s_4$ & $\cdot$ & $\cdot$ 			& $i_2$ & $\cdot$ & $i_1$ & $i_1$ & $i_1$ &  $i_5$& $\cdot$\\
$s_3$ & $\cdot$ &$\cdot$ & $\cdot$ & $s_1$ & $\cdot$ & $\cdot$ 				& $\cdot$ & $\cdot$ & $\cdot$ & $i_6$ & $\cdot$ &  $i_1$&$\cdot$ \\
$s_5$ & $\cdot$ & $\cdot$ & $\cdot$ & $\bm{s_5}$ & $\cdot$ & $\cdot$ 			& $\cdot$ & $\cdot$ & $\cdot$ & $i_5$ & $\cdot$ & $\cdot$ & $\cdot$ \\
$\bm{s_1}$ & $\cdot$ & $\cdot$ & $\cdot$ & $\cdot$ & $\cdot$ & $\cdot$ 			& $\cdot$ & $\cdot$ & $\cdot$ & $\cdot$ & $\cdot$ & $\cdot$ & $\cdot$ 
\end{tabular}
\end{center}

The DA envy digraph appears below, showing the trading cycles implemented by EADA with full consent (Figure 1), improving 4 students through the cycle $(i_1 \to i_6 \to i_4 \to i_5 \to i_1)$ and an alternative, larger improvement (Figure 2), benefiting 6 students via two disjoint cycles, namely $(i_1 \to i_2 \to i_1)$ and $(i_3 \to i_6 \to i_4 \to i_5 \to i_3)$. 
Interestingly, EADA with full consent generates 3 blocking pairs ($\{i_3,s_6\},\{i_5,s_6\},\{i_7,s_4\}$), whereas the alternative, larger improvement generates only two blocking pairs ($\{i_1,s_4\},\{i_7,s_4\}$).\footnote{Note that here we are simply counting the number of blocking pairs; there is a literature (see Section \ref{sec:literature}) arguing that that some are more relevant than others.} Note that the larger improvement benefits every student who could benefit under some allocation that Pareto-dominates DA, because $i_7$ cannot trade with anyone since no student envies him.

\begin{figure}[h]
\centering
		\begin{minipage}{0.45\textwidth}
	\centering
	\begin{tikzpicture}[scale=0.7, ->, node distance=2.5cm, every node/.style={draw, circle, minimum size=1cm}, thick]
		\def\radius{3}
		
		\foreach \i/\name in {1/i_1, 2/i_2, 3/i_3, 4/i_4, 5/i_5, 6/i_6, 7/i_7} {
			\node[circle, draw] (\name) at ({360/7 * (\i-1)}:\radius) {$\name$}; 
		}
		
		\foreach \i in {1,4,5,6} {
			\node[circle, draw=blue, minimum size=1cm] (i_\i) at ({360/7 * (\i-1)}:\radius) {$i_\i$}; 
		}
		
		\draw[ultra thick, blue] (i_1) to (i_6);
		\draw[thick] (i_1) to (i_4);
		\draw[thick] (i_1) to (i_2);
		\draw[thick] (i_1) to (i_3);
		\draw[thick] (i_1) to (i_5);
		
		\draw[thick] (i_2) to (i_1);
		
		\draw[thick] (i_3) to (i_6);
		
		\draw[ultra thick, blue, bend right] (i_4) to (i_5);
		
		\draw[ultra thick, blue, bend right=10] (i_5) to (i_1);
		\draw[thick] (i_5) to (i_4);
		\draw[thick] (i_5) to (i_6);
		\draw[thick] (i_5) to (i_3);
		
		\draw[ultra thick, blue] (i_6) to (i_4);
		
		\draw[thick] (i_7) to (i_4);
	\end{tikzpicture}
	\caption{EADA, full consent.}
\end{minipage}
\hfill
\begin{minipage}{0.45\textwidth}
	\centering
	\begin{tikzpicture}[scale=0.7, ->, node distance=2.5cm, every node/.style={draw, circle, minimum size=1cm}, thick]
		\def\radius{3}
		
		\foreach \i/\name in {1/i_1, 2/i_2, 3/i_3, 4/i_4, 5/i_5, 6/i_6, 7/i_7} {
			\node[circle, draw] (\name) at ({360/7 * (\i-1)}:\radius) {$\name$}; 
		}
		
		\foreach \i in {1,2,3,4,5,6} {
			\node[circle, draw=red, minimum size=1cm] (i_\i) at ({360/7 * (\i-1)}:\radius) {$i_\i$}; 
		}
		
		\draw[thick] (i_1) to (i_6);
		\draw[thick] (i_1) to (i_4);
		\draw[ultra thick, red] (i_1) to (i_2);
		\draw[thick] (i_1) to (i_3);
		\draw[thick] (i_1) to (i_5);
		
		\draw[ultra thick, red] (i_2) to (i_1);
		
		\draw[ultra thick, red] (i_3) to (i_6);
		
		\draw[ultra thick, red, bend right] (i_4) to (i_5);
		
		\draw[thick, bend right=10] (i_5) to (i_1);
		\draw[thick] (i_5) to (i_4);
		\draw[thick] (i_5) to (i_6);
		\draw[ultra thick, red] (i_5) to (i_3);
		
		\draw[ultra thick, red] (i_6) to (i_4);
		
		\draw[thick] (i_7) to (i_4);
	\end{tikzpicture}
	\caption{Larger improvement.}
\end{minipage}
\end{figure}
 To conclude the proof, consider the mechanism $M$ that chooses the same matching as EADA with full consent in every problem except the aforementioned one. Such matching clearly is Pareto-efficient, Pareto-dominates DA, improves more students than EADA with full consent and is cardinally more stable than EADA.

{\it Comparison with DA+TTC.} 	To show that a larger improvement can generate fewer blocking pairs than DA+TTC, consider the following example with five students and five schools with unit capacity. DA's allocation appears in bold.
\begin{center}
\begin{tabular}{ccccc|ccccc}
	\multicolumn{10}{c}{Example 3} \\
	\hline
$i_1$ & $i_2$ & $i_3$ & $i_4$ & $i_5$ & $s_1$ & $s_2$ & $s_3$ & $s_4$ & $s_5$ \\
\hline
$s_2$ & $s_1$ & $s_1$ & $s_1$ & $s_1$ & $\bm{i_1}$ & $\bm{i_2}$ & $\bm{i_3}$ & $\bm{i_4}$ & $\bm{i_5}$ \\
$s_4$ & $s_3$ & $s_2$ & $s_2$ & $\bm{s_5}$ & $i_5$ & $i_4$ & $i_4$ & $i_5$ & $\cdot$ \\
$\bm{s_1}$ & $\bm{s_2}$ & $\bm{s_3}$ & $\bm{s_4}$ & $\cdot$ & $i_4$ & $i_3$ & $i_2$ & $i_1$ & $\cdot$ \\
$\cdot$ & $\cdot$ & $\cdot$ & $\cdot$ & $\cdot$ & $i_2$ & $i_1$ & $\cdot$ & $\cdot$ & $\cdot$ 
\end{tabular}
\end{center}

In DA+TTC, students $i_1$ and $i_2$ trade their DA allocations, generating four blocking pairs: $(i_4,s_1), (i_5, s_1), (i_3, s_2), (i_4,s_2)$. In contrast, the largest improvement is the cycle $(i_1 \to i_4 \to i_1), (i_2 \to i_3 \to i_2)$, generating only one blocking pair: $(i_5,s_1)$. 

To conclude the proof, consider the mechanism $M$ that chooses the same matching as DA+TTC in every problem except the aforementioned one. Such matching clearly is Pareto-efficient, Pareto-dominates DA, improves more students than DA+TTC and is cardinally more stable than DA+TTC.
\end{proof}

We now discuss some implications of Theorem \ref{prop:prop2}. First, our two examples reveal an interesting asymmetry: in Example 2, DA+TTC achieves the larger improvement that doubly dominates EADA, whereas in Example 3, EADA with full consent produces the improvement that double dominates DA+TTC. This observation yields the following incomparability result:

\begin{corollary}
\label{cor:incomparability}
EADA with full consent and DA+TTC are incomparable in terms of size of improvement and number of blocking pairs, i.e. there exist problems where DA+TTC doubly dominates EADA and vice versa.
\end{corollary}

We find Corollary \ref{cor:incomparability} useful because DA+TTC has been considered suboptimal regarding several axioms \citep{kesten2010school,troyan2020essentially}, yet it is not consistently doubly dominated by the more well-regarded EADA mechanism.

Second, the possibility of simultaneously achieving broader student coverage and fewer blocking pairs suggests that alternative efficient mechanisms that Pareto-dominate DA merit investigation. In particular, understanding the properties of such mechanisms—their computational complexity, laboratory performance, and incentive properties—represents an exciting area for future research. While EADA has deservedly received wide praise in the literature, our results point to unexplored opportunities within the class of mechanisms that Pareto-dominate DA that could consistently outperform EADA in this particular criterion of more beneficiaries with fewer blocking pairs.

Finally, we discuss the relationship between the impossibility result (Theorem \ref{thm:impo}) and the possibility result (Theorem \ref{prop:prop2}). Both setwise and cardinal minimal instability capture legitimate concerns about mechanism design, though they reflect different policy goals. Setwise minimal instability focuses on the diversity of potential challenges to the allocation, recognizing that each blocking pair represents a distinct grievance that could motivate legal or political opposition. Cardinal minimal instability, by contrast, adopts a more utilitarian perspective that treats each blocking pair as imposing similar costs regardless of the specific students involved. While our impossibility result applies specifically to setwise minimally unstable mechanisms, this does not invalidate the setwise criterion but rather illuminates the fundamental trade-offs it entails. How much this limitation matters in practice depends critically on policymaker objectives: those prioritizing broad political consensus and legal robustness may find the setwise criterion compelling despite its constraints on student coverage, while those focused on maximizing total improvements or minimizing aggregate unfairness may prefer cardinal comparisons. 


\section{Conclusion}
\label{sec:conclusion}

Our analysis establishes that setwise minimally unstable mechanisms, including EADA, face fundamental constraints on student coverage. We prove these mechanisms achieve the worst possible improvement ratio, potentially helping only $2$ students when $n-1$ could benefit. However, this limitation disappears under cardinal comparisons of blocking pairs: alternative mechanisms can simultaneously improve more students than EADA while generating fewer blocking pairs.

Our findings reveal that the choice between setwise and cardinal minimal instability represents an important design decision. While setwise minimally unstable mechanisms offer protection against diverse legal challenges, it does so at the cost of significantly reduced student participation in efficiency gains. The relatively underexplored class of mechanisms that Pareto-dominate DA contains promising alternatives for policymakers willing to prioritize broader student coverage over setwise stability constraints.
\newpage
\setlength{\parskip}{-0.1em} 
\singlespacing      
\bibliographystyle{te} 
\bibliography{bibliogr}

\begin{thebibliography}{33}
\newcommand{\enquote}[1]{``#1''}
\providecommand{\natexlab}[1]{#1}
\providecommand{\url}[1]{\texttt{#1}}
\providecommand{\urlprefix}{URL }
\providecommand{\bibAnnoteFile}[1]{%
  \IfFileExists{#1}{\begin{quotation}\noindent\textsc{Key:} #1\\
  \textsc{Annotation:}\ \input{#1}\end{quotation}}{}}
\providecommand{\bibAnnote}[2]{%
  \begin{quotation}\noindent\textsc{Key:} #1\\
  \textsc{Annotation:}\ #2\end{quotation}}

\bibitem[{Abdulkadiroglu et~al.(2005)Abdulkadiroglu, Pathak, and
  Roth}]{abdulkadiroglu2005}
Abdulkadiroglu, Atila, Parag Pathak, and Alvin Roth (2005), \enquote{The new
  york city high school match.} \emph{American Economic Review}, 95(2),
  364--367.
\bibAnnoteFile{abdulkadiroglu2005}

\bibitem[{Abdulkadiro{\u{g}}lu et~al.(2009)Abdulkadiro{\u{g}}lu, Pathak, and
  Roth}]{abdulkadirouglu2009strategy}
Abdulkadiro{\u{g}}lu, Atila, Parag~A Pathak, and Alvin~E Roth (2009),
  \enquote{Strategy-proofness versus efficiency in matching with indifferences:
  Redesigning the nyc high school match.} \emph{American Economic Review},
  99(5), 1954--78.
\bibAnnoteFile{abdulkadirouglu2009strategy}

\bibitem[{Abdulkadiro{\u{g}}lu and S{\"o}nmez(2003)}]{abdulkadirouglu2003}
Abdulkadiro{\u{g}}lu, Atila and Tayfun S{\"o}nmez (2003), \enquote{School
  choice: A mechanism design approach.} \emph{American Economic Review}, 93(3),
  729--747.
\bibAnnoteFile{abdulkadirouglu2003}

\bibitem[{Abdulkadiroğlu et~al.(2020)Abdulkadiroğlu, Che, Pathak, Roth, and
  Tercieux}]{aeri}
Abdulkadiroğlu, Atila, Yeon-Koo Che, Parag~A. Pathak, Alvin~E. Roth, and
  Olivier Tercieux (2020), \enquote{Efficiency, justified envy, and incentives
  in priority-based matching.} \emph{American Economic Review: Insights}, 2(4),
  425--42.
\bibAnnoteFile{aeri}

\bibitem[{Alcalde and Romero{-}Medina(2017)}]{alcalde2017}
Alcalde, Jos\'e and Antonio Romero{-}Medina (2017), \enquote{Fair student
  placement.} \emph{Theory and Decision}, 83(2), 293--307.
\bibAnnoteFile{alcalde2017}

\bibitem[{Bando(2014)}]{bando2014existence}
Bando, Keisuke (2014), \enquote{On the existence of a strictly strong nash
  equilibrium under the student-optimal deferred acceptance algorithm.}
  \emph{Games and Economic Behavior}, 87, 269--287.
\bibAnnoteFile{bando2014existence}

\bibitem[{Bonkoungou and Nesterov(2025)}]{bonkoungou2025reforms}
Bonkoungou, Somouaoga and Alexander Nesterov (2025), \enquote{When do reforms
  meet fairness concerns in school admissions?} \emph{Social Choice and
  Welfare}, 1--25.
\bibAnnoteFile{bonkoungou2025reforms}

\bibitem[{Cerrone et~al.(2024)Cerrone, Hermstrüwer, and
  Kesten}]{cerrone2022school}
Cerrone, Claudia, Yoan Hermstrüwer, and Onur Kesten (2024), \enquote{{School
  choice with consent: An experiment}.} \emph{The Economic Journal}.
\bibAnnoteFile{cerrone2022school}

\bibitem[{Che and Tercieux(2019)}]{che2019efficiency}
Che, Yeon-Koo and Olivier Tercieux (2019), \enquote{Efficiency and stability in
  large matching markets.} \emph{Journal of Political Economy}, 127(5),
  2301--2342.
\bibAnnoteFile{che2019efficiency}

\bibitem[{Chen and M{\"o}ller(2023)}]{chen2023regret}
Chen, Yiqiu and Markus M{\"o}ller (2023), \enquote{Regret-free truth-telling in
  school choice with consent.} \emph{Theoretical Economics}.
\bibAnnoteFile{chen2023regret}

\bibitem[{Cheng et~al.(2024)Cheng, Wang, and Yu}]{cheng4961755equilibrium}
Cheng, Yao, Dazhong Wang, and Jingsheng Yu (2024), \enquote{Equilibrium
  analysis of the efficiency-adjusted deferred acceptance mechanism.}
  \emph{Available at SSRN 4961755}.
\bibAnnoteFile{cheng4961755equilibrium}

\bibitem[{Do{\u{g}}an and Ehlers(2021{\natexlab{a}})}]{dougan2021minimally}
Do{\u{g}}an, Battal and Lars Ehlers (2021{\natexlab{a}}), \enquote{Minimally
  unstable pareto improvements over deferred acceptance.} \emph{Theoretical
  Economics}, 16(4), 1249--1279.
\bibAnnoteFile{dougan2021minimally}

\bibitem[{Do{\u{g}}an and Ehlers(2021{\natexlab{b}})}]{ehlers2021robust}
Do{\u{g}}an, Battal and Lars Ehlers (2021{\natexlab{b}}), \enquote{Robust
  minimal instability of the top trading cycles mechanism.} \emph{American
  Economic Journal: Microeconomics}.
\bibAnnoteFile{ehlers2021robust}

\bibitem[{Dogan and Ehlers(2023)}]{dogan2023existence}
Dogan, Battal and Lars Ehlers (2023), \enquote{Existence of myopic-farsighted
  stable sets in matching markets.} In \emph{Proceedings of the 24th ACM
  Conference on Economics and Computation}, 536--536.
\bibAnnoteFile{dogan2023existence}

\bibitem[{Do{\u{g}}an and Yenmez(2020)}]{dougan2020consistent}
Do{\u{g}}an, Battal and M~Bumin Yenmez (2020), \enquote{Consistent pareto
  improvement over the student-optimal stable mechanism.} \emph{Economic Theory
  Bulletin}, 8, 125--137.
\bibAnnoteFile{dougan2020consistent}

\bibitem[{Dur et~al.(2019)Dur, Gitmez, and Y{\i}lmaz}]{dur2019school}
Dur, Umut, A~Arda Gitmez, and {\"O}zg{\"u}r Y{\i}lmaz (2019), \enquote{School
  choice under partial fairness.} \emph{Theoretical Economics}, 14(4),
  1309--1346.
\bibAnnoteFile{dur2019school}

\bibitem[{Ehlers and Morrill(2020)}]{ehlers2020legal}
Ehlers, Lars and Thayer Morrill (2020), \enquote{(il) legal assignments in
  school choice.} \emph{The Review of Economic Studies}, 87(4), 1837--1875.
\bibAnnoteFile{ehlers2020legal}

\bibitem[{Gale and Shapley(1962)}]{gale1962}
Gale, David and Lloyd~S Shapley (1962), \enquote{College admissions and the
  stability of marriage.} \emph{The American Mathematical Monthly}, 69(1),
  9--15.
\bibAnnoteFile{gale1962}

\bibitem[{Karp(2010)}]{karp2010reducibility}
Karp, Richard~M (2010), \emph{Reducibility among Combinatorial Problems}.
  Springer.
\bibAnnoteFile{karp2010reducibility}

\bibitem[{Kesten(2010)}]{kesten2010school}
Kesten, Onur (2010), \enquote{School choice with consent.} \emph{The Quarterly
  Journal of Economics}, 125(3), 1297--1348.
\bibAnnoteFile{kesten2010school}

\bibitem[{Kesten and Kurino(2019)}]{kesten2019strategy}
Kesten, Onur and Morimitsu Kurino (2019), \enquote{Strategy-proof improvements
  upon deferred acceptance: A maximal domain for possibility.} \emph{Games and
  Economic Behavior}, 117, 120--143.
\bibAnnoteFile{kesten2019strategy}

\bibitem[{Kwon and Shorrer(2020)}]{kwon2020justified}
Kwon, Hyukjun and Ran~I Shorrer (2020), \enquote{Justified-envy-minimal
  efficient mechanisms for priority-based matching.} \emph{Available at SSRN
  3495266}.
\bibAnnoteFile{kwon2020justified}

\bibitem[{Ortega and Klein(2023)}]{ortega2023cost}
Ortega, Josu{\'e} and Thilo Klein (2023), \enquote{The cost of
  strategy-proofness in school choice.} \emph{Games and Economic Behavior},
  141, 515--528.
\bibAnnoteFile{ortega2023cost}

\bibitem[{Ortega et~al.(2024)Ortega, Ziegler, Arribillaga, and
  Zhao}]{ortega2024unimprovable}
Ortega, Josu\'e, Gabriel Ziegler, R~Pablo Arribillaga, and Geng Zhao (2024),
  \enquote{Identifying and quantifying (un)improvable students.} \emph{arXiv
  preprint arXiv:2407.19831}.
\bibAnnoteFile{ortega2024unimprovable}

\bibitem[{Ortega et~al.(2025)Ortega, Ziegler, Arribillaga, and
  Zhao}]{ortega2025pareto}
Ortega, Josue, Gabriel Ziegler, R~Pablo Arribillaga, and Geng Zhao (2025),
  \enquote{What pareto-efficiency adjustments cannot fix.} \emph{arXiv preprint
  arXiv:2506.11660}.
\bibAnnoteFile{ortega2025pareto}

\bibitem[{Reny(2022)}]{reny2022efficient}
Reny, Philip~J (2022), \enquote{Efficient matching in the school choice
  problem.} \emph{American Economic Review}, 112(6), 2025--43.
\bibAnnoteFile{reny2022efficient}

\bibitem[{Roth(1982)}]{roth1982economics}
Roth, Alvin~E (1982), \enquote{The economics of matching: Stability and
  incentives.} \emph{Mathematics of operations research}, 7(4), 617--628.
\bibAnnoteFile{roth1982economics}

\bibitem[{Shapley and Scarf(1974)}]{shapley1974}
Shapley, Lloyd and Herbert Scarf (1974), \enquote{On cores and indivisibility.}
  \emph{Journal of Mathematical Economics}, 1(1), 23 -- 37.
\bibAnnoteFile{shapley1974}

\bibitem[{Shirakawa(2025)}]{shirakawa2024simple}
Shirakawa, Ryo (2025), \enquote{Simple manipulations in school choice
  mechanisms.} \emph{American Economic Journal: Microeconomics}.
\bibAnnoteFile{shirakawa2024simple}

\bibitem[{Tang and Yu(2014)}]{tang2014new}
Tang, Qianfeng and Jingsheng Yu (2014), \enquote{A new perspective on kesten's
  school choice with consent idea.} \emph{Journal of Economic Theory}, 154,
  543--561.
\bibAnnoteFile{tang2014new}

\bibitem[{Tang and Zhang(2021)}]{tang2021weak}
Tang, Qianfeng and Yongchao Zhang (2021), \enquote{Weak stability and pareto
  efficiency in school choice.} \emph{Economic Theory}, 71, 533--552.
\bibAnnoteFile{tang2021weak}

\bibitem[{Troyan et~al.(2020)Troyan, Delacr{\'e}taz, and
  Kloosterman}]{troyan2020essentially}
Troyan, Peter, David Delacr{\'e}taz, and Andrew Kloosterman (2020),
  \enquote{Essentially stable matchings.} \emph{Games and Economic Behavior},
  120, 370--390.
\bibAnnoteFile{troyan2020essentially}

\bibitem[{Troyan and Morrill(2020)}]{troyan2020obvious}
Troyan, Peter and Thayer Morrill (2020), \enquote{Obvious manipulations.}
  \emph{Journal of Economic Theory}, 185, 104970.
\bibAnnoteFile{troyan2020obvious}

\end{thebibliography}

\end{document}